\documentclass[showpacs,amsmath,amssymb,twocolumn,pra,superscriptaddress]{revtex4-1}

\usepackage[dvips]{graphicx} 
\usepackage{amsfonts}
\usepackage{amssymb}
\usepackage{amscd}
\usepackage{amsmath}    
\usepackage{enumerate}
\usepackage{epsfig}
\usepackage{bm}
\usepackage{xcolor}
\usepackage{amsthm}
\usepackage{framed}
\usepackage{multirow}
\usepackage{mathrsfs,amssymb}
\usepackage{latexsym} 
\usepackage{framed}
\usepackage{braket}
\usepackage{epstopdf}
\usepackage{verbatim}
\usepackage{pgf}

\newtheorem{theorem}{Theorem}
\newtheorem{lemma}{Lemma}
\newtheorem{definition}{Definition}

\newcommand{\setfigurewidth}[3]{%
  \pgfmathsetmacro{\figwidth}{#2/#3}
  \includegraphics[width=\figwidth\textwidth]{#1}
}
\begin{document}

\title{Numerical Security Analysis of Three-State Quantum Key Distribution Protocol with Realistic Devices}
\date{\today}

\begin{abstract}
Quantum key distribution (QKD) is a secure communication method that utilizes the principles of quantum mechanics to establish secret keys. The central task in the study of QKD is to prove security in the presence of an eavesdropper with unlimited computational power. In this work, we successfully solve a long-standing open question of the security analysis for the three-state QKD protocol with realistic devices, i,e, the weak coherent state source.
 We prove the existence of the squashing model for the measurement settings in the three-state protocol. This enables the reduction of measurement dimensionality, allowing for key rate computations using the numerical approach. We conduct numerical simulations to evaluate the key rate performance. The simulation results show that we achieve a communication distance of up to 200 km. 


\end{abstract}
\author{Sirui Peng}
\affiliation{State Key Lab of Processors, Institute of Computing Technology, Chinese Academy of Sciences, 100190, Beijing, China.}

\author{Xiaoming Sun}
\affiliation{State Key Lab of Processors, Institute of Computing Technology, Chinese Academy of Sciences, 100190, Beijing, China.}

\author{Hongyi Zhou}
\email{zhouhongyi@ict.ac.cn}
\affiliation{State Key Lab of Processors, Institute of Computing Technology, Chinese Academy of Sciences, 100190, Beijing, China.}

\maketitle

\section{Introduction}

Quantum key distribution (QKD) is one of the most promising secure communication methods in the era of quantum technologies. It is a central research direction in quantum information science and quantum cryptography \cite{scarani2009security,xu2020secure}, enabling communication partners, Alice and Bob, to share a secret key while being under the presence of an eavesdropper who has complete control over the quantum channel and infinite computational power. This concept is known as information-theoretic security. QKD leverages the unique properties of quantum mechanics, such as the no-cloning theorem and the uncertainty principle, to establish a shared secret key. This is a fundamental difference from classical cryptography, which relies on the computational complexity of mathematical problems for security.

Since the proposal of the first quantum key distribution protocol in 1984 \cite{BB84}, various protocols have been introduced to improve performance or simplify protocol designs. Among them, the three-state protocol \cite{PBC00} is particularly noteworthy. It extends the B92 protocol \cite{B92} by introducing a third non-orthogonal quantum state. These three quantum states form an equilateral triangle on the $X$-$Z$ plane of the Bloch sphere. Similar to B92 protocol, the measurement can only be described as a general positive operator-valued measure (POVM) rather than a random choice of projective measurements. Compared to the original B92 protocol, the three-state protocol allows for a broader range of eavesdropper detection and enhances the key rate performance. Currently, there is only an unconditional security proof for the three-state protocol with single-photon source settings \cite{Boileau2004UnconditionalSO}. However, implementing an ideal single-photon source in practical scenarios is quite challenging. Instead, weak coherent state sources are commonly employed as an approximation of an ideal single-photon source. Establishing rigorous security proof for coherent state source settings remains an open problem.

A viable approach to proving the security of QKD protocols with realistic devices is the numerical approach, initially proposed in \cite{Winick2018ReliableNK}. Compared to conventional analytical approaches, the numerical approach is more versatile for different protocol designs, particularly for protocols lacking symmetry \cite{coles2016numerical}. Key rate calculations are formulated as convex optimization problems, which can be efficiently solved. However, a major obstacle in the numerical approach is the dimensionality issue. In practical implementations of QKD protocols, the use of coherent state sources involves infinite-dimensional measurements, rendering computations intractable. The squashing model \cite{squashing} has been introduced to address this issue for certain protocols by reducing the real protocol with coherent state sources to a qubit-based protocol. However, the existence of the squashing model in the three-state protocol remains unknown.

In this work, we successfully solve the long-standing open question regarding the security analysis of the  three-state protocol with realistic devices. 
Through the utilization of the numerical approach, we have reformulated the security analysis of this protocol into a convex optimization problem. Our starting point was the well-studied scenario involving an ideal single-photon source. Our analysis reveals that higher bit error tolerance can be achieved compared to previous results obtained using analytical approaches.
Our main contribution lies in the establishment of a robust security proof for the case of a weak coherent state source. We have proven the existence of a squashing model, which allows us to reduce the measurement dimension and compute the key rate using the numerical approach. Considering a standard optical fiber loss of 0.2 dB/km, our result shows that the communication distance can achieve over 200 km. By leveraging the numerical approach and addressing the challenges associated with realistic implementations, we can further enhance the security, efficiency, and practicality of QKD protocols.


\section{Preliminary}
\subsection{Numerical approach of security analysis}
The central task of quantum key distribution is security analysis, where we need to concern about the maximum information leakage in the given protocol against any possible attack compatible with quantum mechanics. The security of the quantum key distribution protocol has a more detailed expression in terms of the key generation rate, i.e., the rate at which the communication parties can generate a shared secret key. We focus on the asymptotic security analysis throughout this paper, where the number of exchanged signals goes to infinity.
The asymptotic key rate formula is given by an optimization problem \cite{Devetak2005DistillationOS},
\begin{equation}\label{equ:keyrate}
\begin{aligned}
    K & = \left(\min _{\rho \in \mathbf{S}} f(\rho)\right)-p_{\text {pass}} \cdot \mathrm{leak}_{\mathrm{obs}}^{\mathrm{EC}} \\
      & =  \left(\min _{\rho \in \mathbf{S}} D(\mathcal{G}(\rho)\|\mathcal{Z}(\mathcal{G}(\rho)))\right)-p_{\text {pass}} \cdot \mathrm{leak}_{\mathrm{obs}}^{\mathrm{EC}} ,
\end{aligned}
\end{equation}
where $D(\cdot\|\cdot)$ is the quantum relative entropy function, $\mathcal{G}$ and $\mathcal{Z}$ are superoperators that respectively describe the post-selection and key-map process in the protocol (whose detailed introductions can be found in Appendix~\ref{sec:pcd}),
$p_{\text{pass}}$ denotes the probability for passing post-selection, and $\mathrm{leak}_{\mathrm{obs}}^{\mathrm{EC}}$ refers to bits revealed for error correction.

Next, we explain the variable $\rho$, which is a bipartite state defined from the entanglement-based protocol.
In the entanglement-based protocol, Alice prepares a bipartite state $\ket{\Psi}_{AA^\prime}$ in a single round,
\begin{equation}\label{equ:Alice-preparation-genreal}
\ket{\Psi}_{AA^\prime}=\sum_{j}\sqrt{p_j}\ket{j}_A\ket{\psi_j}_{A^\prime},
\end{equation}
where $\ket{\psi_j}$ stands for signal states and $\{p_j\}_j$ stands for its probability distribution. Then the measurement on system $A$ equals to a random preparation of $\ket{\psi_j}$.
Alice sends system $A^\prime$ to Bob through an untrusted channel $\mathcal{E}$, which contains both channel noise and possible eavesdropping. After the channel, the bipartite state becomes 
\begin{equation}\label{equ:channel-on-AAprime}
\rho_{AB} = (I\otimes\mathcal{E})(\rho_{AA^\prime}),
\end{equation}
which is exactly the variable in Eq.~\eqref{equ:keyrate}. For simplicity, we will neglect the subscript in the optimization problem in the following context. 

Finally, we explain the feasible set $\mathbf{S}$. In the entanglement-based protocol, Alice and Bob perform some joint measurements on $\rho_{AB}$. The corresponding observables are denoted as
 $\{\Gamma_i\}_{i}$ whose expectations are $\{\gamma_i\}_{i}$. The expectations together with some state preparation information such as the overlap between different signal states, serve as the constraints of the optimization problem. As we will see later, the state preparation information can also be expressed as expectations of some observables. The feasible set $\mathbf{S}$ is the set of the bipartite quantum state $\rho$ satisfying the constraints, i.e., $\mathbf{S}=\{\rho \succeq 0| \operatorname{Tr}\left(\Gamma_i \rho\right)=\gamma_i, \forall i\}$.
Then we can optimize $f(\rho)$ the further calculate the key rate by Eq.~\eqref{equ:keyrate}. 

To get reliable lower bounds of key rate, a two-phase optimization is applied in \cite{Winick2018ReliableNK}. In the first step, the Frank-Wolfe method is used to get a sub-optimal solution $\rho^\prime$, which provides an upper bound on key rate. The second step converts the upper bound into a lower bound by solving the dual problem of the linearization of function $f$ around the specified point $\rho^\prime$. Since the dual problem is a maximization problem, the lower bound remains reliable even if the numerical computation does not attain the global optimum.


\subsection{Symmetric three-state protocol}\label{sec:protocol}
We briefly review the symmetric three-state protocol proposed in \cite{PBC00}, which is known as the PBC00 protocol. In the protocol, the three quantum signal states are encoded as $\ket{\psi_j}=e^{\frac{2j\pi i}{3} \sigma_y}\ket{0}$ $(j\in\{0,1,2\})$ in the polarization of a single photon, where $\sigma_y$ is the Pauli $Y$ matrix. These three states form an equilateral triangle on the $X$-$Z$ plane. The prepare-and-measure PBC00 protocol works as follows,
\begin{enumerate}
    \item Alice generates a random bit $b\in\{0,1\}$ and a trit $r\in\{0,1,2\}$. Then she prepares the signal state $\ket{\psi_{f(r,b)}}$. The function $f(r,b)$ is defined in Table~\ref{tab:step1_mapping}.
    \begin{table}[htbp]
        \centering
        \begin{tabular}[t]{r|r r}
            \hline
             $f(r,b)$ & $b=0$       & $b=1$     \\
            \hline
             $r=0$    & $0$    & $1$  \\
             $r=1$    & $1$    & $2$  \\
             $r=2$    & $2$    & $0$  \\         
            \hline
        \end{tabular}
        \caption{Value table of function $f$.}
        \label{tab:step1_mapping}
    \end{table}
    \item Bob conducts a measurement on system $B$ (system $A^\prime$ through a noisy channel), which is characterized by a set of positive operator-valued measure (POVM) $\{P_k = 2\ket{\overline{\psi_k}}\bra{\overline{\psi_k}}/3\}_k$ $(k=0,1,2)$, where $\ket{\overline{\psi_k}}$ is the quantum state orthogonal to $\ket{\psi_k}$ on the $X$-$Z$ plane.
    \item After repeating steps 1 and 2 for $N_{\operatorname{tot}}$ rounds, Bob announces that all his measurements have finished. Then Alice announces the values of the random trit in each round.
    \item For each round, Bob maps measurement result $p$ and the announced trit $r$ to a secret key, which is characterized by a function $g(p,r)$, and announces the indices of inconclusive results. The function of $g(p,r)$ is defined in Table~\ref{tab:step4_mapping}.
    \begin{table}[htbp]
        \centering
        \begin{tabular}[t]{r|r r r}
            \hline
            $g(p,r)$ & $r=0$ & $r=1$ & $r=2$ \\
            \hline
            $p=0$   & $1$   & $X$   & $0$   \\
            $p=1$   & $0$   & $1$   & $X$   \\
            $p=2$   & $X$   & $0$   & $1$   \\
            \hline
        \end{tabular}
        \caption{Value table of function $g$, where the inconclusive measurement results in $g$ are marked as $X$.}
        \label{tab:step4_mapping}
    \end{table}
    \item Alice discards inconclusive bits in $b$ and selects some of string $b$ as test string (which does not affect key rate in asymptotic case). If estimated bit error rate is in an acceptable range, Alice and Bob perform error correction and privacy amplification to get the final secret key. 
\end{enumerate}
We focus on the asymptotic case where $N_{\operatorname{tot}}\rightarrow \infty$.
\section{Result}
In this section, we present the security analysis of the three-state protocol with settings given in Fig.~\ref{fig:loss_channel}.
To start with, we consider an ideal case where a single-photon source and ideal single-photon detectors are applied. We formulate the optimization problem by the numerical framework and simulate the key rate by assuming a depolarizing channel.
Next, we consider the case with realistic settings, where the phase-randomized weak coherent state source is applied as an approximation of the ideal single-photon source and detected by threshold detectors with dark counts. To overcome the problem of infinite dimensional signal state and measurement in numerical computation, we prove a squashing map exists for the measurements in the symmetric three-state protocol. Then again we formulate the optimization problem and simulate the key rate by assuming a loss channel.

\begin{figure*}[htbp]
    \centering
    \setfigurewidth{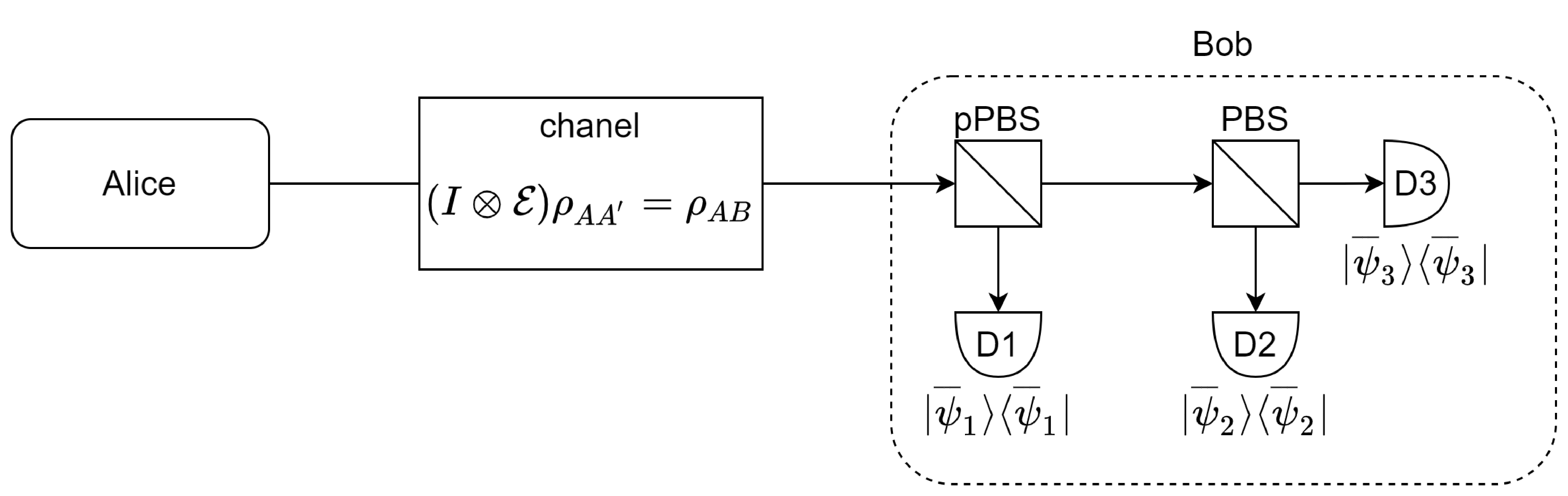}{2193}{3000}
    \caption{Overview of the three-state protocol. Bob's measurement settings are described by the POVM $\{P_k\}_k$. The click of the threshold detector $D_k$ corresponds to the POVM element $P_k$. PBS: polarization beam splitter. pPBS: partial polarization beam splitter.}
    \label{fig:loss_channel}
\end{figure*}

\subsection{Single-photon source case}
\subsubsection{Formulation of the optimization problem}
To formulate the optimization problem, we need to set the dimension of the variable $\rho$ and formulate the observables involved in the constraints. We consider the entanglement-based protocol where Alice prepares a bipartite entangled state $\ket{\Psi}_{AA^\prime}$,
\begin{equation}
\ket{\Psi}_{AA^\prime}=\sum_{r,b}\sqrt{\frac{1}{6}}\ket{r}_{A_1}\ket{b}_{A_2}\ket{\psi_j}_{A^\prime},
\end{equation}
In the description of the three-state protocol in Sec.~\ref{sec:protocol}, the quantum signal prepared in each round is determined by $r$ and $b$. Then random choice of $r$ and $b$ is equal to a measurement on two quantum registers in Hilbert space $\mathcal{H}_{A_1}$ and $\mathcal{H}_{A_2}$, respectively. They together form Alice's ancillary system $A$, i.e., $\mathcal{H}_{A}=\mathcal{H}_{A_1}\otimes \mathcal{H}_{A_2}$, with dimension $2\times3=6$. Since the quantum signal is a qubit, the dimension of $\rho$ is 12.

The constraints can be divided into two categories. One is that the expectations of the observables involved in the actual experiment should be compatible with the experimental statistics. In the entanglement-based protocol, the observables are described by a bipartite Hermitian operator $O_{j,k}$, $j\in\{0,1,2\}$. The observable $O_{j,k}$ corresponds to the event where Alice sends $j$-th signal state and Bob's measurement output is $k$,
\begin{equation}
 O_{j,k}=Q_j\otimes P_k,
\end{equation}
where 
\begin{equation}\label{constraint_union_probability}
\begin{aligned}
 &Q_0=\ket{00}\bra{00}_A+\ket{21}\bra{21}_A \\
& Q_1=\ket{01}\bra{01}_A+\ket{10}\bra{10}_A \\
& Q_2=\ket{11}\bra{11}_A+\ket{20}\bra{20}_A.
\end{aligned}    
\end{equation}
The other category of constraints reveals the information on state preparation,
\begin{equation}\label{constraint_inner_prod}
    \operatorname{Tr}_B(\rho_{AB})=\operatorname{Tr}_{A^\prime}(\rho_{AA^\prime}) =\rho_A,
\end{equation}
which is independent of Eve's attack. It can be rewritten in terms of expectations of some observables,
\small
\begin{equation}
\Theta_{r,b;r^\prime,b^\prime} = 
\begin{cases}
\frac{1}{2}(\ket{r,b}\bra{r^\prime,b^\prime}+\ket{r^\prime,b^\prime}\bra{r,b}), & 2r^\prime + b^\prime \leq 2r + b \\
\frac{i}{2}(\ket{r^\prime,b^\prime}\bra{r,b}-\ket{r,b}\bra{r^\prime,b^\prime}), & 2r^\prime + b^\prime > 2r + b
\end{cases},
\end{equation}
\normalsize
where $r,r^\prime\in\{0,1,2\},b,b^\prime\in\{0,1\}$.
Then we are able to compute the following optimization problem,
\begin{equation}
\begin{aligned}
\text{variable}\quad & \rho
\\
\min\quad & f(\rho)
\\
\mathrm{s.t.}\quad & \rho \succeq 0\\
& \operatorname{Tr}(O_{j,k}\rho) = o_{j,k} & \forall j,k
\\
& \operatorname{Tr}((\Theta_{r,b;r^\prime,b^\prime}\otimes I)\rho) = \theta_{r,b;r^\prime,b^\prime} & \forall r,r^\prime,b,b^\prime,
\end{aligned}
\end{equation}
where the variable is a $12\times 12$ positive semi-definite matrix, $j\in\{0,1,2\}$, $k\in\{0,1,2\}$, $r,r^\prime \in\{0,1,2\}$, and $b,b^\prime\in\{0,1\}$.
The expectations $\theta_{r,b;,r^\prime,b^\prime}$ are determined by the state preparation information,
\small
\begin{equation}
\theta_{r,b;r^\prime,b^\prime} = 
\begin{cases}
\frac{1}{6}\operatorname{Re}\braket{\psi_{f(r,b)}|\psi_{f(r^\prime,b^\prime)}}, & 2r^\prime + b^\prime \leq 2r + b \\
\frac{1}{6}\operatorname{Im}\braket{\psi_{f(r,b)}|\psi_{f(r^\prime,b^\prime)}}, & 2r^\prime + b^\prime > 2r + b
\end{cases},
\end{equation}
\normalsize
where $o_{j,k}$ is the probability that Bob's measurement result is $k$ and Alice sends $j$-th state, which should be obtained from experiment.

\subsubsection{Simulation result}
For the key rate calculation, we make a simulation of $o_{j,k}$ by assuming a depolarizing channel, which is described as
\begin{equation}
\mathcal{E}(\rho)=(1-p)\rho + p\frac{I}{d},
\end{equation}
where $d$ is the dimension of $\rho$ and $p$ is the parameter characterizing the depolarization.
With this channel model, we are able to simulate $o_{j,k}=\operatorname{Tr}((I\otimes \mathcal{E})(\rho_{AA^\prime})O_{j,k})$.


To calculate the final key rate in Eq.~\eqref{equ:keyrate}, we also need to simulate the probability of passing post-selection $p_{\text {pass}}$ and the bit error rate $e_{\text{bit}}$,
\begin{equation}
\begin{aligned}
p_{\text {pass}} &=\frac{1}{2}\sum_{i\neq j} \left( \frac{1-p}{3}\braket{\overline{\psi_j} | \psi_i}^2 + \frac{p}{6}\braket{\overline{\psi_j} | \psi_i} \right)
\\
e_{\text{bit}} &= \sum_{i} \left( \frac{1-p}{3}\braket{\overline{\psi_i} | \psi_i}^2 + \frac{p}{6}\braket{\overline{\psi_i} | \psi_i} \right)/p_{\text {pass}}.
\end{aligned}
\end{equation}
Our simulation result is shown in Figure \ref{fig:bachelor-4-1}. It turns out that the numerical approach can lead to a higher bit error tolerance compared with the analytical one \cite{Boileau2004UnconditionalSO}.

\begin{figure}[htbp]
    \centering
    \includegraphics[width=0.5\textwidth]{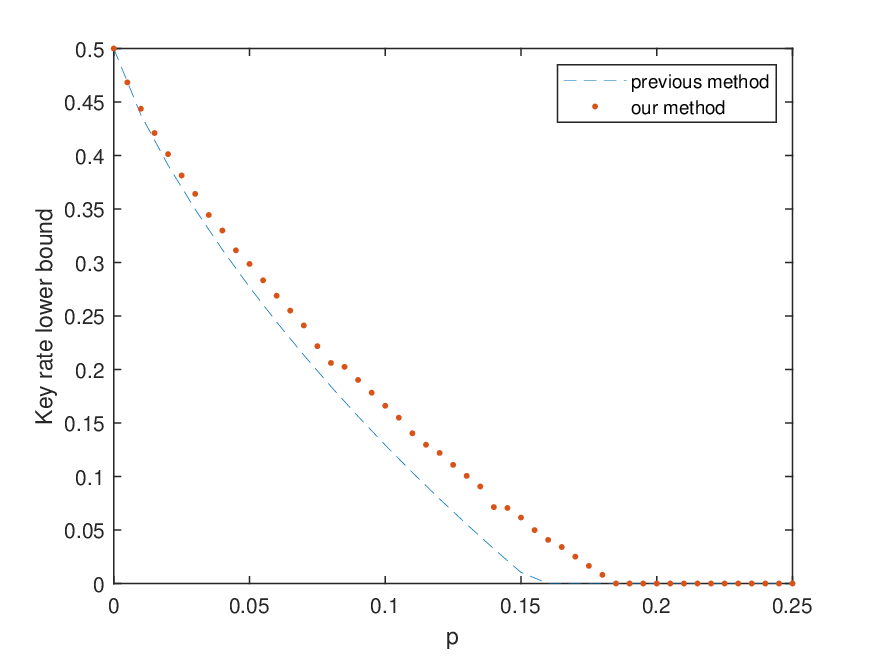}
    \caption{Simulation results on the symmetric three-state protocol's key rate versus the channel parameter $p$. Our results show that a non-zero key rate can be obtained for $p < 0.18$, corresponding to a quantum bit error rate of $0.113$, which outperforms the previous result ($p < 0.152$, $e_{\text{bit}} < 0.098$) \cite{Boileau2004UnconditionalSO}.}
    \label{fig:bachelor-4-1}
\end{figure}

\subsection{Coherent-state source case}
The original three-state protocol is established under the assumption of single-photon source. In this section, we consider the more practical weak coherent state source $\ket{\alpha}$, which is a superposition of Fock state $\ket{n}$, i.e., $\ket{\alpha}=e^{-|\alpha|^2/2}\sum_{n=0}^\infty \alpha^n\ket{n}/\sqrt{n!}$. After the phase randomization, it becomes a mixture of Fock states with a Poisson distribution $p_\mu(n)$,
\begin{equation}\label{equ:cohernet_source}
\rho_j(\mu) = \sum_{n=0}^\infty p_\mu(n)\ket{n}\bra{n} = \sum_{n=0}^\infty \frac{\mu^n e^{-\mu}}{n!}\ket{n}\bra{n},
\end{equation}
where $\mu=|\alpha|^2$.
In the polarization encoding, $\ket{n}$ refers to $n$ photons in the same polarization state $\ket{\psi_j}$ $(j\in\{0,1,2\})$.


\subsubsection{Squashing model}\label{sec:squashing-model}
The depiction of the coherent source as well as the measurement involves infinite dimensions, which makes a direct numerical computation impossible. Fortunately, the squashing model \cite{squashing} provides an approach to reduce the coherent state source scheme to qubit-based one. 
We define some notations: Bob's qubit measurement is characterized by a set of POVM $F_Q=\{F_Q^{(k)}\}_k$, and his measurement of threshold detectors on the practical signal is characterized by the basic measurement $F_B=\{F_B^{(k^\prime)}\}_{k^\prime}$. After some classical post-processing, $F_B$ is transformed into a coarse-grained full measurement $F_M=\{F_M^{(k)}\}_k$, such that the number of POVM elements in $F_M$ is the same as that of $F_Q$. The transformation of classical post-processing is represented by a stochastic matrix $\mathcal{P}$. The squashing model is described as follows.
\begin{figure*}[htbp]
    \centering
    \setfigurewidth{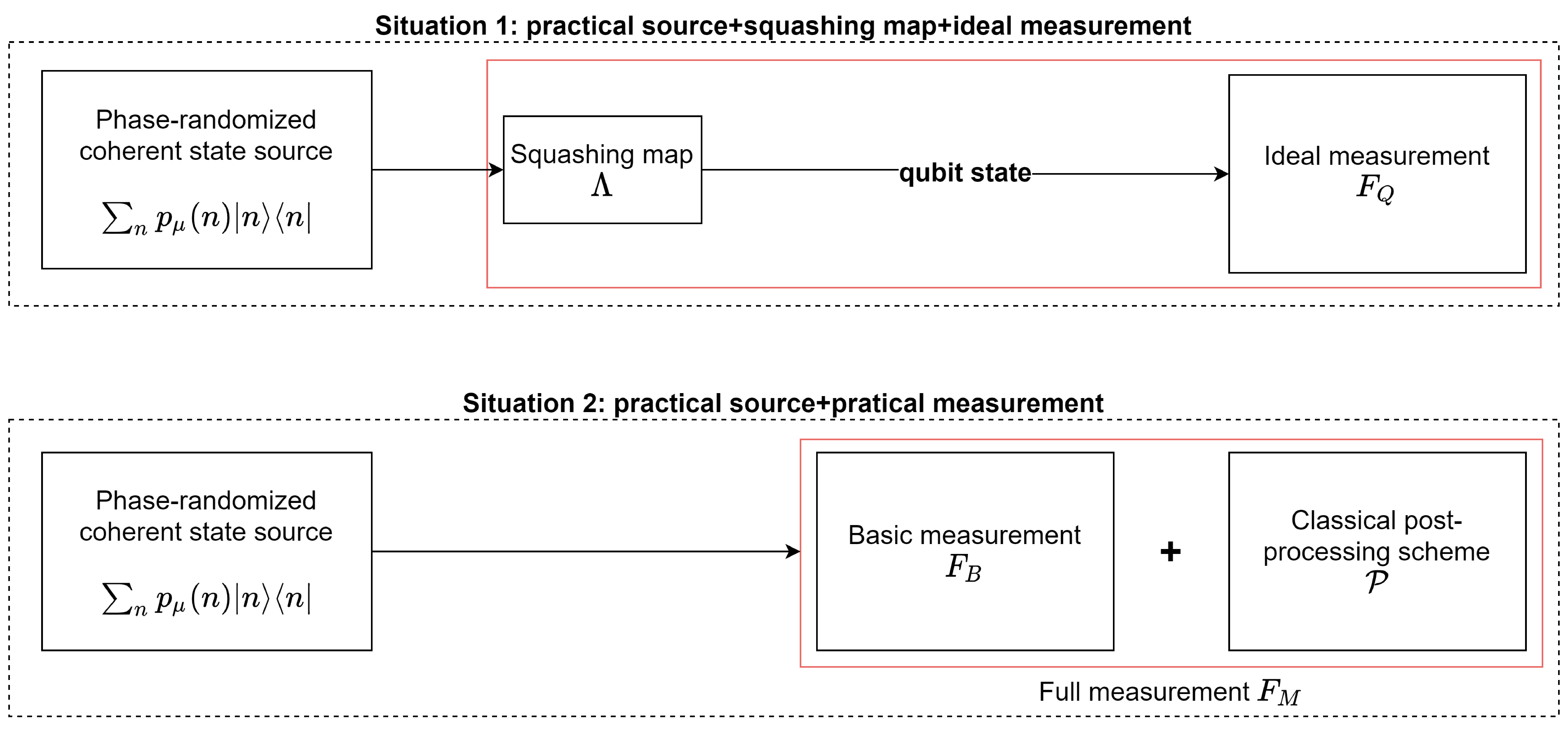}{3804}{4000}
    \caption{Schematics of the squashing model. The input state is first mapped to a qubit state, and then measured by an ideal qubit measurement. If a squashing model exists, the red frames in the two situations are equivalent. }
    \label{fig:squashing_model_scheme}
\end{figure*}

\begin{definition}
(Squashing model) Let $F_B=\{F_B^{(1)},F_B^{(2)}\cdots F_B^{(l^\prime)}\}$ and $F_Q=\{F_Q^{(1)},F_Q^{(2)}\cdots F_Q^{(l)}\}$ be the POVMs that describe the basic and target measurements in high dimensional and low dimensional Hilbert spaces, $\mathcal{H}_M$ and $\mathcal{H}_Q$, respectively. Let $\mathcal{P}$ be a stochastic matrix corresponding to the classical post-processing scheme, which leads to a full measurement POVM $F_M$ based on $F_B$. Then for a given $\mathcal{P}$, a squashing model for the measurement settings exists if there exists a map $\Lambda:\mathcal{H}_M\to\mathcal{H}_Q$ such that
\begin{enumerate}
    \item For any input state $\rho_M\in\mathcal{H}_M$,
\begin{equation}\label{equ:squashing-constraints}
\operatorname{Tr}\left(F_M^{(j)} \rho_M\right)=\operatorname{Tr}\left(F_Q^{(j)} \Lambda\left[\rho_M\right]\right), \forall j.
\end{equation}
    \item $\Lambda$ is a completely positive and trace-preserving (CPTP) map.
\end{enumerate}
\end{definition}
By the definition above, if there is such a map, the security analysis for coherent source protocol can be reduced to that for a qubit-based protocol.

For linear optical measurement device with threshold detectors, the squashing map can be decomposed with respect to the photon number subspaces,
\begin{equation}\label{eq:reduction1}
\Lambda[\rho] \stackrel{\text { QND }}{=}\Lambda\left[\bigoplus_{N=0}^{\infty} \rho_N\right]=\bigoplus_{N=0}^{\infty} \Lambda_N\left[\rho_N\right].
\end{equation}
This is because the basic measurement commutes with the QND measurement of the total photon number. Then we can consider a QND measurement on the input state $\rho_M$ first, which dephases $\rho_M$ into block-diagonal form in Fock basis, $\rho_M= \bigoplus_{N=0}^{\infty} \rho_{M,N}$. The constraints in Eq.~\eqref{equ:squashing-constraints} still holds for squashing map on $N$-photon subspace,
\begin{equation}
\operatorname{Tr}\left(F_{M,N}^{(j)} \rho_{M,N}\right)=\operatorname{Tr}\left(F_Q^{(j)} \Lambda_N\left[\rho_{M,N}\right]\right), \forall j.
\end{equation}
This virtual QND measurement allows us to split off the vacuum component both in dephased input states and measurements. We let $\Lambda_0(\rho_{M,0})=\ket{\text{vac}}\bra{\text{vac}}$, where $\ket{\text{vac}}$ is the vacuum state. That is to say, the squashing map always outputs a vacuum state when a zero-photon outcome is obtained in the QND measurement. Then Eq.~\eqref{eq:reduction1} is rewritten as
\begin{equation}
\Lambda[\rho]=\ket{\text{vac}}\bra{\text{vac}}\oplus\bigoplus_{N=1}^{\infty} \Lambda_N\left[\rho_N\right].
\end{equation}
Our following discussion is restricted to the subspace of $N\in \mathbb{N}_+$ photon number, and vacuum components in the target and full measurement are ignored.


To describe the squashing model in the symmetric three-state protocol, we introduce the following definitions and notations. We use $\ket{N_1,N_2,N_3}$ to denote a state where $N_i$ photons are in the polarization state $\ket{\psi_i}$.
Another notation $\ket{\bar N_1,\bar N_2,\bar N_3}$ is similarly defined, where there are $N_i$ of the photons in the polarization state $\ket{\overline{\psi_i}}$.
The target measurement $\textbf{F}_{Q} = [F_{Q}^{(1)},F_{Q}^{(2)},F_{Q}^{(3)}]^T$ becomes a qubit measurement after ignoring the vacuum component, which is equivalent to $P_k$,
\begin{equation}\label{equ:target_measurement}
\begin{aligned}
F_{Q}^{(1)} &= \frac{2}{3}\ket{\bar 1,\bar 0,\bar 0}\bra{\bar 1,\bar 0,\bar 0}
\\
F_{Q}^{(2)} &= \frac{2}{3}\ket{\bar 0,\bar 1,\bar 0}\bra{\bar 0,\bar 1,\bar 0}
\\
F_{Q}^{(3)} &= \frac{2}{3}\ket{\bar 0,\bar 0,\bar 1}\bra{\bar 0,\bar 0,\bar 1}.
\end{aligned}
\end{equation}
The basic measurement $\textbf{F}_{B}=[F_{B}^{(1)},F_{B}^{(2)},\dotsc,F_{B}^{(1,2,3)}]^T$ depends on the property of threshold detectors that can only distinguish zero and non-zero photons,
\small
{\allowdisplaybreaks
\begin{align}\label{equ:basic_measurement}
F_{B}^{(1)}&= \bigoplus_{N>0}\left(\frac{2}{3}\right)^N\ket{\bar N,\bar 0,\bar 0}\bra{\bar N,\bar 0,\bar 0} \nonumber
\\
F_{B}^{(2)}&= \bigoplus_{N>0}\left(\frac{2}{3}\right)^N\ket{\bar 0,\bar N,\bar 0}\bra{\bar 0,\bar N,\bar 0} \nonumber
\\
F_{B}^{(3)}&= \bigoplus_{N>0}\left(\frac{2}{3}\right)^N\ket{\bar 0,\bar 0,\bar N}\bra{\bar 0,\bar 0,\bar N}
\\
F_{B}^{(1,2)}&=\bigoplus_{N>0}\sum_{\substack{N_1,N_2>0 \\ N_1+N_2=N}} \left(\frac{2}{3}\right)^N\ket{\bar N_1,\bar N_2,\bar 0}\bra{\bar N_1,\bar N_2,\bar 0} \nonumber
\\
F_{B}^{(1,3)}&=\bigoplus_{N>0}\sum_{\substack{N_1,N_2>0 \\ N_1+N_2=N}} \left(\frac{2}{3}\right)^N\ket{\bar N_1,\bar 0,\bar N_2}\bra{\bar N_1,\bar 0,\bar N_2} \nonumber
\\
F_{B}^{(2,3)}&=\bigoplus_{N>0}\sum_{\substack{N_1,N_2>0 \\ N_1+N_2=N}} \left(\frac{2}{3}\right)^N\ket{\bar 0,\bar N_1,\bar N_2}\bra{\bar 0,\bar N_1,\bar N_2} \nonumber
\\
F_{B}^{(1,2,3)}&=I-\sum_{i=1,2,3} F^{(i)}_{B}-\sum_{\substack{i\neq j \\ i=1,2,3 \\j=1,2,3}} F^{(i,j)}_{B}. \nonumber
\end{align}
}
\normalsize
We choose a post-processing scheme that randomly distributes double-click and triple-click events to each measurement outcomes,
\begin{equation}\label{equ:full_measurement}
\begin{aligned}
F_{M}^{(1)}&=F_{B}^{(1)}+\frac{1}{3}(F_{B}^{(1,2)}+F_{B}^{(1,3)}+F_{B}^{(2,3)}+F_{B}^{(1,2,3)})
\\
F_{M}^{(2)}&=F_{B}^{(2)}+\frac{1}{3}(F_{B}^{(1,2)}+F_{B}^{(1,3)}+F_{B}^{(2,3)}+F_{B}^{(1,2,3)})
\\
F_{M}^{(3)}&=F_{B}^{(3)}+\frac{1}{3}(F_{B}^{(1,2)}+F_{B}^{(1,3)}+F_{B}^{(2,3)}+F_{B}^{(1,2,3)}).
\end{aligned}
\end{equation}
Combining the basic measurement with post-processing, we obtain the coarse-grained full measurement $\textbf{F}_{M}=[F_{M}^{(1)},F_{M}^{(2)},F_{M}^{(3)}]^T$. The vectorized POVM $\textbf{F}_{M}$ and $\textbf{F}_{B}$ are connected by a linear transformation,
\begin{equation}
\textbf{F}_{M} = \mathcal{P} \textbf{F}_{B},
\end{equation}
where the stochastic matrix $\mathcal{P}$ is
\begin{equation}\label{equ:stochastic-matrix}
\mathcal{P} = \begin{bmatrix}
1 & 0 & 0 &
\frac{1}{3} & \frac{1}{3} & \frac{1}{3} & \frac{1}{3}
\\
0 & 1 & 0 &
\frac{1}{3} & \frac{1}{3} & \frac{1}{3} & \frac{1}{3}
\\
0 & 0 & 1 &
\frac{1}{3} & \frac{1}{3} & \frac{1}{3} & \frac{1}{3}
\end{bmatrix}.
\end{equation}
By the definitions above, we prove that the squashing model does exist for the measurement settings of the symmetric three-state protocol.
\begin{theorem}\label{thm:squashing-model-existence}
The squashing model for the symmetric three-state protocol exists, i.e., for the given target measurement Eq.~\eqref{equ:target_measurement} and full measurement Eq.~\eqref{equ:full_measurement}, there exists a CPTP map $\Lambda$ such that 
\begin{equation}
\operatorname{Tr}(\rho F_M^{(k)}) =\operatorname{Tr}(\Lambda(\rho) F_Q^{(k)}), \forall \rho, k.
\end{equation}
\end{theorem}
We briefly introduce the sketch of the proof and leave the detailed proof in Appendix \ref{sec:smn}. We follow the two reductions in \cite{ squashing,gittsovich2014squashing} to simplify the proof. The first reduction is based on the observation of virtual QND measurement mentioned above. By this reduction, we only need to prove the existence of the squashing model in each photon-number subspace.
The second reduction is to further divide the $N$-photon subspace into the single-click subspace $\Lambda_{P,N}$ and its orthogonal compliment subspace. The proof is further reduced to finding a squashing map limited on $N$-photon single-click subspace. 
The general proof idea is to find the Choi matrix for the squashing map in each $N$-photon single-click subspace and then check whether the Choi matrix is positive definite. We discuss the cases of $N\leq 5$ and $N>5$. In the first case, there is a trivial identity map for $N=1$. We specify the squashing maps for $N\in\{2,3,4,5\}$ and verify them by numerically calculating the minimum eigenvalue. While in the second case, we notice that the off-diagonal terms of the Choi matrix of $\Lambda_{P,N}$ decay exponentially with $N$, i.e., the Choi matrix of $\Lambda_{P,N}$ converges to the identity matrix when $N\rightarrow \infty$. Then we prove the convergence by applying the technique of Gershgorin's circle \cite{Ger31}.
  

\subsubsection{Formulation of the optimization problem}
Compared with the ideal case, the only difference in the formulation lies in Bob's POVM which is denoted as $\{P^\prime_k\}_k$. Based on the vacuum flag structure of
the squashing map \cite{squashing}, an extra dimension is introduced in Bob's POVM space to include vacuum subspace, $P^\prime_k = 0 \oplus P_k$ $(k\in\{0,1,2\})$. There is also an additional POVM element corresponding to the no-click event, $P_3^\prime = I-\sum_{k=0}^2 P^\prime_k$. Then $O^\prime_{j,k} = Q_j\otimes P^\prime_k$.

According to \cite{Wang2022NumericalSP}, we consider the contribution from the single-photon component as a lower bound of the key rate,
\begin{equation}\label{eq:key_rate_coherent}
\begin{aligned}
K &= \left(\sum_{n=0}^{\infty} p_\mu(n) \min_{\rho^{(n)}\in S_n} D(\mathcal{G}(\rho^{(n)})\|\mathcal{Z}(\mathcal{G}(\rho^{(n)})))\right) \\
&\quad - p_{\mathrm{pass}} \cdot \mathrm{leak}_{\mathrm{obs}}^{\mathrm{EC}}
\\
&\geq p_\mu(1) \min_{\rho^{(1)}\in \textbf{S}_1} f(\rho^{(1)}) - p_{\mathrm{pass}} \cdot \mathrm{leak}_{\mathrm{obs}}^{\mathrm{EC}},
\end{aligned}
\end{equation}
where $\rho^{(n)}$ stands for $n$-photon component, $p_{\mu}(n)$ is its probability following the Poisson distribution, and $\textbf{S}_1$ is the feasible set of the single photon component $\rho^{(1)}$. To calculate Eq.~\eqref{eq:key_rate_coherent}, we formulate the following optimization problem whose feasible set is exactly $\textbf{S}_1$,
\begin{equation}\label{eq:coherent_opt}
\begin{aligned}
\text{variable}\quad & \rho
\\
\min\quad & f(\rho)
\\
\mathrm{s.t.}\quad &\rho \succeq 0 \\
& \operatorname{Tr}(O^\prime_{j,k}\rho) \in [o^L_{j,k,1},o^U_{j,k,1}] & \forall j,k
\\
& \operatorname{Tr}((\Theta_{r,b;r^\prime,b^\prime}\otimes I)\rho) = \theta_{r,b;r^\prime,b^\prime} & \forall r,r^\prime,b,b^\prime,
\end{aligned}
\end{equation}
where the variable is a $18\times 18$ positive semi-definite matrix, the indices $j\in\{0,1,2\}$, $k\in\{0,1,2,3\}$, $r,r^\prime \in\{0,1,2\}$, and $b,b^\prime\in\{0,1\}$. As a key difference from the single-photon source case, the equality constraints of the expectations are replaced with inequality ones.
We apply the decoy state method \cite{Lo2004DecoySQ} to estimate the bounds $o^L_{j,k,1}$ and $o^U_{j,k,1}$.
The idea of decoy state method is to express $o_{j,k,\mu}$, the probability of obtaining $k$ given $j$-th signal state with coherent state source of amplitude $\mu$, as a linear combination of $o_{j,k,n}$, i.e., statistics contributed by the $n$-photon component,
\begin{equation}
o_{j,k,\mu}=\sum_{n=0}^{\infty} p_{\mu}(n) o_{j,k,n}.
\end{equation}
Then one can estimate the single photon contribution $o_{j,k,1}$ by solving a linear program for finite decoy states.
In our case, the decoy state method is to solve the following linear program,
\small
\begin{equation}\label{eq:decoy_lp}
\begin{aligned}
\text{variable} \quad & o_{j,k,n}
\\
\max/\min \quad & o_{j,k,1} \\
\mathrm{s.t.} \quad & o_{j,k,n}\in [0,1]
&\forall n \\
& o_{j,k,\mu}\leq\sum_{n\leq N} p_{\mu}(n) o_{j,k,n} + (1 - \sum_{n\leq N}p_{\mu}(n))
&\forall \mu \\
& o_{j,k,\mu}\geq\sum_{n\leq N} p_{\mu}(n) o_{j,k,n}
&\forall \mu,
\end{aligned}
\end{equation}
\normalsize
In our calculation, for simplicity, we choose the photon number cut-off of $10$ and decoy states number of $2$, i.e., $\mu \in\{\mu_1,\mu_2, \mu_3\}$ where $\mu_1$ is the intensity of the signal state. 




\subsubsection{Simulation result}
The quantities we need to simulate are $o_{j,k,\mu}$. Then we can apply the decoy state method to estimate $o_{j,k,1}$ by Eq.~\eqref{eq:decoy_lp} and obtain its bounds $o^L_{j,k,1}$ and $o^U_{j,k,1}$ as constraints in Eq.~\eqref{eq:coherent_opt}. We consider a loss channel characterized by the transmittance $\eta$ for coherent state case, i.e., it transforms a coherent state $\ket{\alpha}$ to $\ket{\sqrt{\eta}\alpha}$. The detector dark count $p_d$ is taken into account too. When Alice sends the $j$-th signal state, the intensity of the incoming signal of the $k$-th detector $D_k$ is given by 
\begin{equation}
\mu_{jk} = \frac{2}{3}\eta\mu |\braket{\psi_j|\overline{\psi_k}}|^2.
\end{equation}
Then we can calculate the probability that $D_k$ clicks given that Alice sends $j$-th signal state, $p_{j,k,\mu}=1-e^{-\mu_{jk}}(1-p_d)$, and the expectations of basic measurement outcomes given the $j$-th signal state,
\begin{equation}
\begin{aligned}
o_{B;j,\overline{b_1b_2b_3},\mu} = \prod_{k=1,2,3} \left((2 b_k-1)p_{j,k,\mu}+(1-b_k)\right) \\
b_1,b_2,b_3\in\{0,1\},\overline{b_1b_2b_3}\neq 0,
\end{aligned}
\end{equation}
where we use the binary expression $\overline{b_1b_2b_3}$ to denote the value  $2^2b_1+2^1b_2+2^0b_3$.
Our post-processing scheme Eq.~\eqref{equ:stochastic-matrix} connects it to the expectations of full measurement outcomes,
 \begin{equation}
o_{j,k,\mu} = \sum_{\substack{b_1,b_2,b_3\in\{0,1\}, \\ \overline{b_1b_2b_3}\neq 0}} \mathcal{P}_{k,\overline{b_1b_2b_3}}o_{B;j,\overline{b_1b_2b_3},\mu}.
 \end{equation}
To calculate the final key, we also need to bound the sifting probability and error rate,
\begin{equation}
\begin{aligned}
p_{\text {pass}}\geq p^L_{\text {pass}}
&= \frac{1}{2}(o_{B;0,\overline{100},\mu_1}+o_{B;0,\overline{010},\mu_1}) \\
&+\frac{1}{2}(o_{B;0,\overline{001},\mu_1}+o_{B;0,\overline{100},\mu_1}) \\
&+ \frac{1}{2}(o_{B;1,\overline{010},\mu_1}+o_{B;1,\overline{001},\mu_1}) \\
&+\frac{1}{2}(o_{B;1,\overline{100},\mu_1}+o_{B;1,\overline{010},\mu_1}) \\
&+ \frac{1}{2}(o_{B;2,\overline{001},\mu_1}+o_{B;2,\overline{100},\mu_1}) \\
&+\frac{1}{2}(o_{B;2,\overline{010},\mu_1}+o_{B;2,\overline{001},\mu_1}),
\end{aligned}
\end{equation}
and
\begin{equation}
\begin{aligned}
e_{\text{bit}}\leq e_{\text{bit}}^U= (\sum_j \sum_{\substack{b_1,b_2,b_3\in\{0,1\}, \\ b_j=1}} o_{B;j,\overline{b_1b_2b_3},\mu_1})/p_{\text{pass}}.
\end{aligned}
\end{equation}
We simulate the key rate versus the transmission distance in Fig.~\ref{fig:coherent1}. The maximum transmission distance can achieve 200 km, which is comparable to BB84 protocol.
\begin{figure}[htbp]
    \centering
    \includegraphics[width=0.5\textwidth]{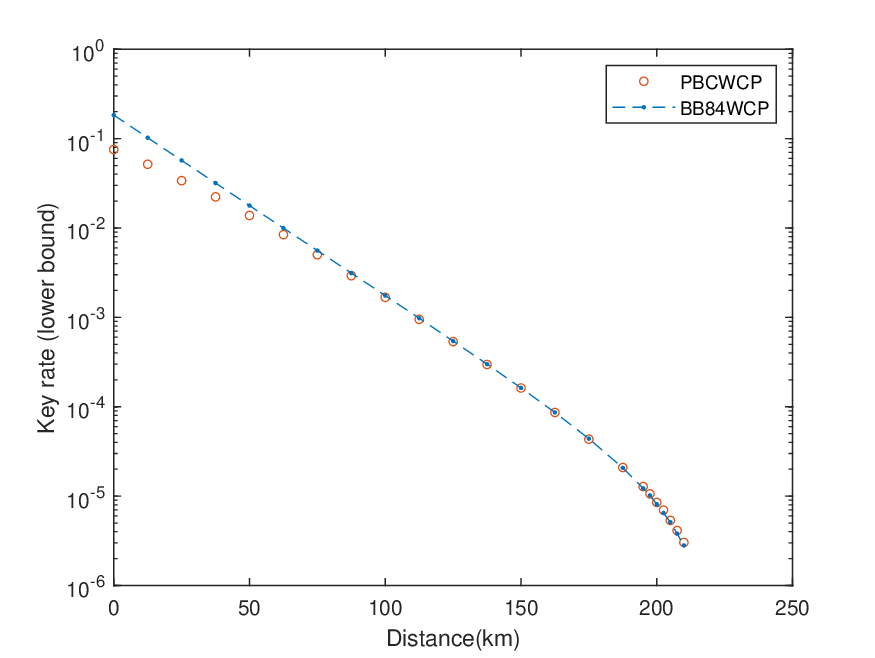}
    \caption{Simulation result of the key rate versus transmission distance for decoy-state PBC00 protocol  and BB84 protocol under the same settings. The fiber loss is set as $0.2$dB/km, and the signal and decoy state intensities are optimized in the range of $\mu_1\in(0.001,1],\mu_2=(0.001,\mu_1],\mu_3\in\{0.001\}$. The probability of dark count is $p_d=10^{-6}$.}
    \label{fig:coherent1}
\end{figure}

\section{Conclusion}
In conclusion, by employing the numerical approach, we have provided rigorous security proof for both single-photon source and coherent state source settings in the PBC00 protocol. In the single-photon source case, our analysis achieves higher bit error tolerance compared to previous results obtained through analytical approaches.

In the coherent state source case, we have successfully proved the existence of the squashing model in the PBC00 protocol. This finding has significant implications not only for the PBC00 protocol but also for other quantum communication protocols involving threshold detectors. The squashing model offers a practical solution to reduce the dimensionality of measurements, enabling efficient computation of the key rate. Our method can be further applied to various prepare-and-measure QKD protocol with general measurements, such as the protocols with passive measurement basis choice.

\section*{Acknowledgement}
This work was supported in part by the National Natural Science Foundation of China Grants No. 61832003, 62272441, 12204489, and the Strategic Priority Research Program of Chinese Academy of Sciences Grant No. XDB28000000.
\bibliographystyle{apsrev4-1}

\bibliography{pbc00ref}
\onecolumngrid
\appendix
\section{Proof of the existence of squashing model}
\label{sec:smn}

To simpilify our search for a valid squashing model, we apply two reductions according to \cite{squashing,gittsovich2014squashing}.
The first reduction decompose the squashing map with respect to the photon number subspace, as introduced in section \ref{sec:squashing-model}.

Moreover, the chosen classical post-processing scheme preserves single-click events in the basic measurement. Then we can apply the second reduction, which states that the $N$-photon squashing map can be further decomposed onto the subspace $P$ spanned by single-click events and its orthogonal complement $P_{\bot}$,
\begin{equation}
\Lambda_N=\Lambda_{P, N}+\Lambda_{P_{\perp}, N}.
\end{equation}
By the post-processing scheme we choose, the latter map $\Lambda_{P_{\perp}, N}$ simply output the maximally mixed state on all inputs.
\begin{figure*}[htbp]
    \centering
    \includegraphics[width=0.75\textwidth]{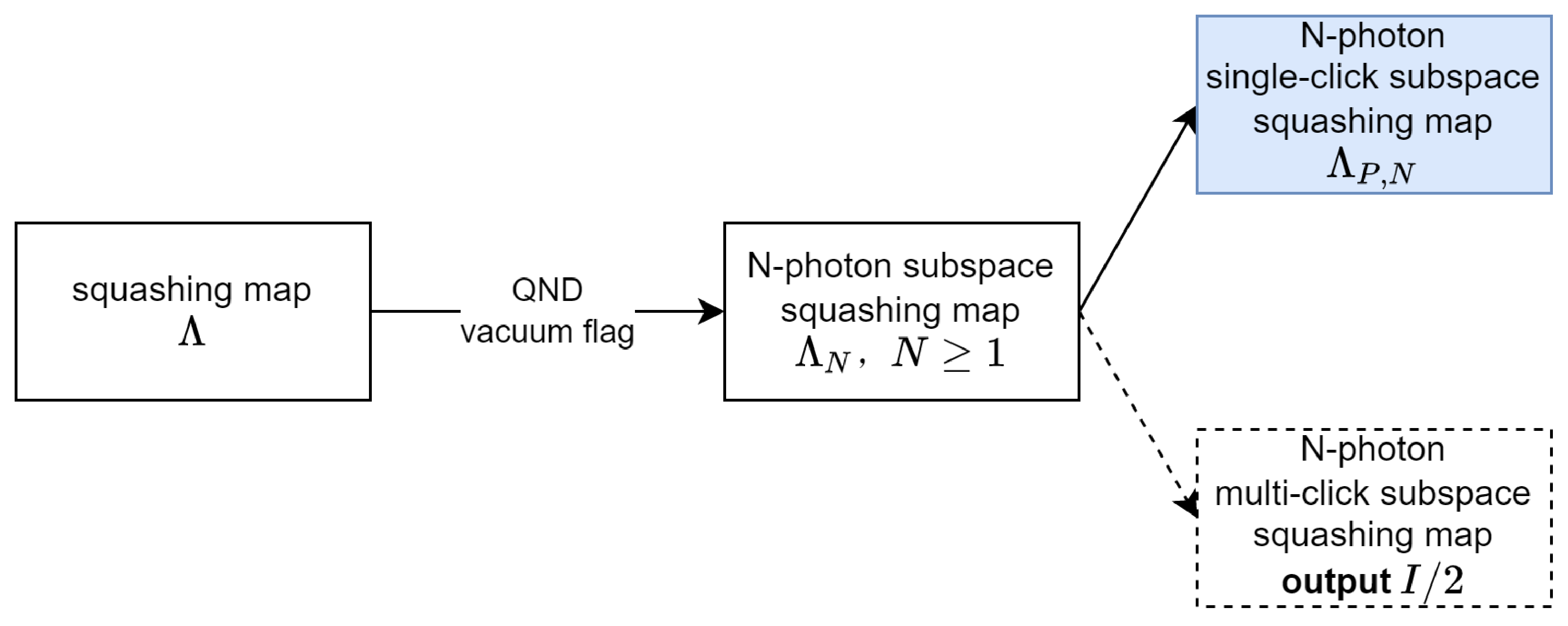}
    \caption{Relationship of $\Lambda$,$\Lambda_N$,$\Lambda_{P,N}$, and $\Lambda_{P_\perp,N}$. Our goal is to find a squashing map $\Lambda$. With reductions 1 and 2 mentioned in \cite{squashing}, the existence of $\Lambda_{P,N}$ for all non-zero photon number subspaces leads to the existence of $\Lambda$. Vacuum projection measurement in the target measurement and full measurement can be ignored due to the vacuum flag design.}
    \label{fig:squashing-reduction}
\end{figure*}
Thus, we have reduced the existence of a squashing map $\Lambda$ to the existence of a squashing map $\Lambda_{P,N}$ for all $N\in\mathbb{N}_+$ on the single-click subspace $P$ spanned by single-click events $\{\ket{\bar N,\bar 0,\bar 0},\ket{\bar 0,\bar N,\bar 0},\ket{\bar 0,\bar 0,\bar N}\}$. In this section, we will prove the latter.
\begin{theorem}
For all $N\in\mathbb{N}_+$, given the target measurement Eq.~\eqref{equ:target_measurement} and the full measurement Eq.~\eqref{equ:full_measurement_n_subspace} on the $N$-photon subspace, there exists a CPTP map $\Lambda_{P,N}$ on the single-click subspace $P$ such that 
\begin{equation}
\operatorname{Tr}(\rho F_{M,N}^{(k)}) =\operatorname{Tr}(\Lambda(\rho) F_{Q}^{(k)}) \; \forall\rho\in P, \forall k.
\end{equation}
\end{theorem}
For the map $\Lambda_{P,N}$, we consider the basic measurement $\textbf{F}_{B,N}=[F_{B,N}^{(1)},F_{B,N}^{(2)},F_{B,N}^{(3)},F_{B,N}^{(1,2)},F_{B,N}^{(1,3)},F_{B,N}^{(2,3)},F_{B,N}^{(1,2,3)}]^T$ and the full measurement $\textbf{F}_{M,N}=[F_{M,N}^{(1)},F_{M,N}^{(2)},F_{M,N}^{(3)}]^T$ in the $N$-photon subspace,
\begin{equation}\label{equ:basic_measurement_n_subspace}
\begin{aligned}
F_{B,N}^{(1)}&= \left(\frac{2}{3}\right)^N\ket{\bar N,\bar 0,\bar 0}\bra{\bar N,\bar 0,\bar 0}
\\
F_{B,N}^{(2)}&= \left(\frac{2}{3}\right)^N\ket{\bar 0,\bar N,\bar 0}\bra{\bar 0,\bar N,\bar 0}
\\
F_{B,N}^{(3)}&= \left(\frac{2}{3}\right)^N\ket{\bar 0,\bar 0,\bar N}\bra{\bar 0,\bar 0,\bar N}
\\
F_{B,N}^{(1,2)}&=\sum_{\substack{N_1,N_2>0 \\ N_1+N_2=N}} \left(\frac{2}{3}\right)^{N_1}\left(\frac{2}{3}\right)^{N_2}\ket{\bar N_1,\bar N_2,\bar 0}\bra{\bar N_1,\bar N_2,\bar 0}
\\
F_{B,N}^{(1,3)}&=\sum_{\substack{N_1,N_2>0 \\ N_1+N_2=N}} \left(\frac{2}{3}\right)^{N_1}\left(\frac{2}{3}\right)^{N_2}\ket{\bar N_1,\bar 0,\bar N_2}\bra{\bar N_1,\bar 0,\bar N_2}
\\
F_{B,N}^{(2,3)}&=\sum_{\substack{N_1,N_2>0 \\ N_1+N_2=N}} \left(\frac{2}{3}\right)^{N_1}\left(\frac{2}{3}\right)^{N_2}\ket{\bar 0,\bar N_1,\bar N_2}\bra{\bar 0,\bar N_1,\bar N_2}
\\
F_{B,N}^{(1,2,3)}&=I_N-\sum_{i=1,2,3} F^{(i)}_{B,N}-\sum_{\substack{i\neq j \\ i=1,2,3 \\j=1,2,3}} F^{(i,j)}_{B,N}.
\end{aligned}
\end{equation}
Then they are connected by a stochastic matrix $\mathcal{P}$,
\begin{equation}\label{equ:full_measurement_n_subspace}
\begin{aligned}
\textbf{F}_{M,N} = \mathcal{P} \textbf{F}_{B,N},
\end{aligned}
\end{equation}
where $\mathcal{P}$ is defined in Eq.~\eqref{equ:stochastic-matrix}.

Next, we want to verify that $\Lambda_{P,N}$ can be a completely positive map under the imposed constraints. 
We consider the adjoint map $\Lambda^\dagger_{P,N}$, which is positive if and only if $\Lambda_{P,N}$ is positive. We use Choi–Jamiołkowski isomorphism to reformulate the positivity of $\Lambda^\dagger_{P,N}$. In order to check for positivity of $\Lambda^\dagger_{P,N}$, we only need to verify that its Choi matrix $\tau_{P,N}$ is positive definite,
\begin{equation}
\begin{aligned}
\tau_{P, N}&=I \otimes \Lambda_{P, N}^{\dagger}\left(\left|\psi^{+}\right\rangle\left\langle\psi^{+}\right|\right)
\\
&=\frac{1}{4}\left(I_{Q} \otimes I_{M}+\sum_{\alpha=x, y, z} \sigma_{\alpha}^{T} \otimes \Lambda_{P, N}^{\dagger}\left(\sigma_{\alpha}\right)\right),
\end{aligned}
\end{equation}
where $\ket{\psi^+}=\frac{1}{\sqrt{2}} \sum_{j=0,1}\ket{j}\ket{j}$ is a maximally entangled state, $I_Q$ and $I_M$ stand for the identity operator on $\mathcal{H}_Q$ and $\mathcal{H}_M$, $\Lambda^{\dagger}_{P,N}$ is the adjoint map of $\Lambda_{P,N}$, and $\sigma_x,\sigma_y,\sigma_z$ stand for the Pauli operators.

We can express $\sigma_x$ and $\sigma_z$ as linear combinations of the POVM elements of the target measurement, while $\sigma_y$ cannot be expressed in this way,
\begin{equation}\label{equ:linear-relation-in-lambda_dagger}
\begin{aligned}
\sigma_z&= F_Q^{(2)} + F_Q^{(3)} - 2F_Q^{(1)}
\\
\sigma_x&= \sqrt{3} (F_Q^{(3)} - F_Q^{(2)})
\\
\sigma_y&\notin \text{span}\{F^{(i)}_Q,i=1,2,3\}.
\end{aligned}
\end{equation}
Considering the adjoint map,
\begin{equation}
\begin{aligned}
\Lambda^\dagger_{P,N}(F^{(j)}_Q) &= F^{(j)}_{M,P,N}\quad j=1,2,3,
\end{aligned}
\end{equation}
where $F^{(j)}_{M,P,N}$ stands for the full measurement POVM element $F^{(j)}_{M}$ restricted $N$-photon single-click subspace $P$\cite{squashing}. We denote $O_p$ as the projection operator on subspace $P$, then
\begin{equation}\label{equ:FMPN_formulation}
\begin{aligned}
F^{(j)}_{M,P,N} 
&= O_p F^{(j)}_{M,N} O_p^\dagger
\\
&= F^{(j)}_{B,N} + \frac{1}{3}O_p(F_{B,N}^{(1,2)}+F_{B,N}^{(1,3)}+F_{B,N}^{(2,3)}+F_{B,N}^{(1,2,3)})O_p^\dagger.
\end{aligned}
\end{equation}
where $O_p=\sum_{i=1,2,3}\ket{\xi_i}\bra{\xi_i}$ and $\{\ket{\xi_i}\}_i$ is a set of orthonormal basis for $P$ obtained from Gram–Schmidt process on $\ket{\bar N,\bar 0,\bar 0},\ket{\bar 0,\bar N,\bar 0},\ket{\bar 0,\bar 0,\bar N}$.
We can see that for $\alpha=x,z$, the image $\Lambda_{P,N}^\dagger(\sigma_\alpha)$ is fixed, while it is not true for $\alpha=y$. We can separate these parts in $\tau_{P,N}$ to clarify the degree of freedom,
\begin{equation}
\begin{aligned}
\tau_{P,N} &= \frac{1}{4}(\tau_{P,N,\text{fixed}}+\tau_{P,N,\text{open}}),
\\
\end{aligned}
\end{equation}
where
\begin{equation}
\begin{aligned}
\tau_{P,N,\text{fixed}} &= 1_Q\otimes 1_{M,N}
\\&
+ \sigma_x^T \otimes \sqrt{3}(F_{M,P,N}^{(3)}-F_{M,P,N}^{(2)})
\\&
+ \sigma_z^T \otimes (F_{M,P,N}^{(2)}+F_{M,P,N}^{(3)}-2F_{M,P,N}^{(1)})
\\
\tau_{P,N,\text{open}} &= \sigma_{y}^{T} \otimes \Lambda_{P, N}^{\dagger}\left(\sigma_{y}\right).
\end{aligned}
\end{equation}
We consider $\mathcal{M}_{\phi}$ that maps a linear operator to a matrix,
\begin{equation}
[\mathcal{M}_{\phi}(A)]_{jk} = \braket{\phi_j|A|\phi_k},
\end{equation}
where $\ket{\phi_j}$ is a set of normal (and not necessarily orthogonal) basis. It is known that positivity of $\mathcal{M}_{\phi}(A)$ implies positivity of $A$ and vice versa. Therefore, we have the following lemma.
\begin{lemma}
For any basis $\{\ket{\phi_j}\}_j$ of an $n$-dimensional Hilbert space $\mathcal{H}$, an operator $A$ on that space is positive if and only if $\mathcal{M}_{\phi}(A)$ is positive definite.
\end{lemma}

\begin{proof}
Consider a set of orthonormal basis $\{\ket{\phi^\prime_j}\}_j$. It is known that an operator is positive if and only if its matrix under a set of orthonormal basis is positive definite, i.e. $A$ is positive if and only if $\mathcal{M}_{\phi^\prime}(A)$ is positive definite. Positivity of $\mathcal{M}_{\phi}(A)$ is equivalent to that for all non-zero vector $b$ in $\mathbb{C}^n$,
$$
b^\dagger \mathcal{M}_{\phi}(A) b >0
$$
then the positivity of $\mathcal{M}_{\phi}(A)$ is equivalent to that of $\mathcal{M}_{\phi^\prime}(A)$, i.e., for all non-zero vector $a$ in $\mathbb{C}^n$,
$$
a^\dagger \mathcal{M}_{\phi^\prime}(A) a=\braket{s|A|s}=b^\dagger \mathcal{M}_{\phi}(A) b>0,\forall a\in \mathbb{C}^n\setminus\{0\},
$$
where $\sum_j a_j\ket{\phi_j}=\ket{s}=\sum_j b_j\ket{\phi^\prime_j}$.

By proving that $\mathcal{M}_{\phi^\prime}(A)$ is positive definite if and only if $\mathcal{M}_{\phi}(A)$ is positive definite, we have proved the lemma.
\end{proof}
Now, let's consider $M_{\Phi}(\tau_{P,N})$, where $\{\ket{\phi_j}\}_j$ is a set of basis vectors on the $N$-photon single-click subspace,
\begin{equation}
\begin{aligned}
\ket{\phi_{1}} &= \ket{\bar N,\bar 0,\bar 0}
\\
\ket{\phi_{2}} &= \ket{\bar 0,\bar N,\bar 0}
\\
\ket{\phi_{3}} &= \ket{\bar 0,\bar 0,\bar N},
\end{aligned}
\end{equation}
and the full basis $\{\ket{\Phi_j}\}_j$ is
\begin{equation}
\begin{aligned}
\ket{\Phi_{jk}} &= \ket{j}\otimes\ket{\phi_{k}},j=0,1,k=1,2,3.
\end{aligned}
\end{equation}
We can represent $\mathcal{M}_{\Phi}(\tau_{P,N})$ in blocks,
\begin{equation}\label{equ:main_choi_matrix}
\mathcal{M}_{\Phi}(\tau_{P,N}) =
\begin{bmatrix}
M_1+M_z & M_x+S  \\
M_x-S   & M_1-M_z
\end{bmatrix}.
\end{equation}
The block matrices are under the basis $\{\ket{\phi_j}\}_j$,
\begin{equation}
\begin{aligned}
M_1 &= \mathcal{M}_{\phi}(I_M)
\\
M_\alpha &= \mathcal{M}_{\phi}(\Lambda_{N}^\dagger(\sigma_\alpha)),\quad 
 \alpha=x,y,z,
\end{aligned}
\end{equation}
and the variable matrix $S=-iM_y$ represents the degree of freedom from $\Lambda^\dagger_{N}(\sigma_y)$,
\begin{equation}
S = \begin{bmatrix}
0    & x_1  & x_2 \\
-x_1 & 0    & x_3 \\
-x_2 & -x_3 & 0
\end{bmatrix}.
\end{equation}
Next, we need to prove that there exists $x=(x_1,x_2,x_3)$ for all $N\in\mathbb{N}_+$ such that $\mathcal{M}_{\Phi}(\tau_{P,N})$ is positive definite.

\subsection{Case for $N\leq 5$}

In this subsection, we will prove the existence of $x$ for some small $N$ such that $\mathcal{M}_{\Phi}(\tau_{P,N})$ is positive definite.

\begin{lemma}\label{lem:case_N_leq_5}
For $N\leq 5$, there exists $x=(x_1,x_2,x_3)$ such that $\mathcal{M}_{\Phi}(\tau_{P,N})$ is positive definite.
\end{lemma}

\begin{proof}
The proof is to solve a feasibility problem, i.e., to find $x$ such that the minimal eigenvalue of $\mathcal{M}_{\Phi}(\tau_{P,N})$ is positive. Here we try to find a feasible point through an optimization process. We look for $x$ that makes the minimal eigenvalue as large as possible. If we find a candidate $x$ that makes the minimal eigenvalue positive, we proves the lemma.

The optimization problem is solved by the optimizer \verb|scipy.optimize.brute|. For $N\in\{2,3,4,5\}$, the feasible solution $x$ and its corresponding objective value are listed in Table \ref{tab:N_leq_5_x}.
For $N=1$ case, we recall the application of virtual QND measurement mentioned in section \ref{sec:squashing-model}. The target measurement is actually equal to single-photon subspace projection of the full measurement,
\begin{equation}
F_Q^{(j)}=F_{M,1}^{(j)},\forall j.
\end{equation}
Therefore, we let $\Lambda_1$ be an identity map. Then we have proved existence for $N\leq 5$.
\end{proof}
\begin{table}[htbp]
    \centering
    \begin{tabular}{r r r r r}
        \hline
         N &  $x_1$ & $x_2$ & $x_3$ & $\min\operatorname{eig}(\mathcal{M}_{\Phi}(\tau_{P,N}))$\\
        \hline
            2 & -0.0047 & 0.0053 &-0.0047 & 0.536 \\
            3 & -0.1093 & 0.1093 & 0.1093 & 0.713 \\
            4 & -0.0121 & 0.0111 &-0.0110 & 0.889 \\
            5 & -0.0367 & 0.0366 & 0.0359 & 0.930 \\
        \hline
    \end{tabular}
    \caption{Minimum eigenvalues of $\mathcal{M}_{\phi}(\tau_{P,N})$ for $N\leq 5$. The positive minimum eigenvalues imply the positivity of $\tau_{P,N}$.}
    \label{tab:N_leq_5_x}
\end{table}
\subsection{Case for $N>5$}
We notice that $\lim_{N\to\infty} \mathcal{M}_{\Phi}(\tau_{P,N})=I$ for $x=0$. Based on this observation, we set $x=0$ and split the matrix into two parts: a major part that converges to the identity matrix and an additional part converges to zero when $N\rightarrow \infty$. Our analysis focuses on the major part and excludes the effects of the additional part by using Gershgorin's circle theorem \cite{Ger31}.
\begin{lemma}\label{lem:case_N_greater_than_5}
The choice of $x=0$ makes $\mathcal{M}_{\Phi}(\tau_{P,N})$ positive for all $N>5$.
\end{lemma}
\begin{proof}
For Hermitian matrices, we only present the upper triangular part. 
The matrix representation of an identity operator is
$$
M_1 = \begin{bmatrix}
1 & a & a \\
  & 1 & b \\
  &   & 1 
\end{bmatrix},
$$
where $a=0.5^N$ and $b=(-0.5)^N$.
For single-click events in $F_{B,N}$,
\begin{equation}\label{equ:matrix-form-fbn}
\begin{aligned}
\mathcal{M}_{\phi}(F^{(1)}_{B,N}) &= \left(\frac{2}{3}\right)^N 
\left(\begin{bmatrix}
1 & 0 & 0 \\
  & 0 & 0 \\
  &   & 0
\end{bmatrix}
+
\delta_1
\right),\text{ where }\delta_1 = \begin{bmatrix}
0 & a & a \\
  & a^2 & a^2 \\
  &     & a^2
\end{bmatrix}
\\
\mathcal{777M}_{\phi}(F^{(2)}_{B,N}) &= \left(\frac{2}{3}\right)^N
\left(\begin{bmatrix}
0 & 0 & 0 \\
  & 1 & 0 \\
  &   & 0
\end{bmatrix}+
\delta_2
\right),\text{ where }\delta_2 = \begin{bmatrix}
a^2 & a & ab \\
    & 0 & b \\
    &   & b^2
\end{bmatrix}
\\
\mathcal{M}_{\phi}(F^{(3)}_{B,N}) &= \left(\frac{2}{3}\right)^N
\left(\begin{bmatrix}
0 & 0 & 0 \\
  & 0 & 0 \\
  &   & 1
\end{bmatrix}
+
\delta_3
\right),\text{ where }\delta_3 = \begin{bmatrix}
a^2 & ab & a \\
  & b^2 & b \\
  &   & 0
\end{bmatrix},
\end{aligned}
\end{equation}
and $\|\delta_j\|_{\max} = a$ for all $j\in\{1,2,3\}$, where $\|\cdot\|_{\max}$ is the max norm of a matrix.
Recall $M_x$ and $M_z$ in Eq.~\eqref{equ:main_choi_matrix}. They can be derived from Eq.~\eqref{equ:matrix-form-fbn} and Eq.~\eqref{equ:FMPN_formulation},
\begin{equation}
\begin{aligned}
M_x &= \sqrt{3}\mathcal{M}_{\phi}(F_{M,P,N}^{(3)}-F_{M,P,N}^{(2)})
\\
&= \sqrt{3}\left(\mathcal{M}_{\phi}(F_{B,N}^{(3)})-\mathcal{M}_{\phi}(F_{B,N}^{(2)})\right)
\\
&= \sqrt{3}\left(\frac{2}{3}\right)^N
\left(\begin{bmatrix}
0 & 0 & 0 \\
  & -1 & 0 \\
  &   & 1
\end{bmatrix}
+\delta_x\right), \;
\delta_x = \delta_1-\delta_2
\\
M_z &= \mathcal{M}_{\phi}(-2F_{M,P,N}^{(1)}+F_{M,P,N}^{(2)}+F_{M,P,N}^{(3)})
\\
&= -2\mathcal{M}_{\phi}(F_{B,N}^{(1)})+\mathcal{M}_{\phi}(F_{B,N}^{(2)})+\mathcal{M}_{\phi}(F_{B,N}^{(3)})
\\
&= \left(\frac{2}{3}\right)^N
\left(\begin{bmatrix}
-2 & 0 & 0 \\
 & 1 & 0 \\
 &  & 1
\end{bmatrix}
+\delta_z
\right),\; \delta_z = -2\delta_1+\delta_2+\delta_3.
\end{aligned}
\end{equation}
and $\|\delta_x\|_{\max}\leq 2a$, $\|\delta_z\|_{\max}\leq 4a$.
Gershgorin's circles for $\mathcal{M}_{\Phi}(\tau_{P,N})$ are defined by their center $r_j$ (approximated by $r^\prime_j$ where $|r_j-r^\prime_j|\leq\left(\frac{2}{3}\right)^N\|\delta_z\|$) and radius $R_j$,
\begin{alignat*}{2}
&
r^\prime_1 = 1-2 \left(\frac{2}{3}\right)^N
&&\quad\quad
R_1 \leq R_{c1} = 2a
+\left(\frac{2}{3}\right)^N\left(2\|\delta_z\|\right)
+\sqrt{3}\left(\frac{2}{3}\right)^N(3\|\delta_x\|)
\\
&
r^\prime_2 = 1+\left(\frac{2}{3}\right)^N = r^\prime_3
&&\quad\quad
R_2, R_3 \leq R_{c2} = 2a
+\left(\frac{2}{3}\right)^N(2\|\delta_z\|)
+\sqrt{3}\left(\frac{2}{3}\right)^N\left(1+3\|\delta_x\|\right)
\\
&
r^\prime_4 = 1+2\left(\frac{2}{3}\right)^N
&&\quad\quad
R_4 \leq R_{c1}
\\
&
r^\prime_5 = 1-\left(\frac{2}{3}\right)^N = r^\prime_6
&&\quad\quad
R_5,R_6 \leq R_{c2}.
\end{alignat*}

Next, we need to verify that all Gershgorin's circles are on the plane $\{x|\text{Re}x>0\}$,
$$r_i-R_i>0,\forall i.$$
Firstly we bound the radius,
\begin{equation}
\begin{aligned}
R_i &\leq R_{c2}, \forall i,
\end{aligned}
\end{equation}
then we bound the center $r_i$,
\begin{equation}
\begin{aligned}
r_i
&\geq r^\prime_i - \left(\frac{2}{3}\right)^N\|\delta_z\|
\\
&>r^\prime_1 - \left(\frac{2}{3}\right)^N\|\delta_z\|,\forall i,
\end{aligned}
\end{equation}
and we can bound the leftmost point on all Gershgorin's circles for all $i$,
\begin{equation}
\begin{aligned}
r_i-R_i&>r^\prime_1 - R_{c2} - \left(\frac{2}{3}\right)^N\|\delta_z\|
\\
&= 1-\left(
2^{1-N}+\left(\frac{2}{3}\right)^N\left(2+\sqrt{3}\right)+\left(\frac{1}{3}\right)^N\left(12+6\sqrt{3}\right)
\right)
\\
&:= f(N).
\end{aligned}
\end{equation}
The function $f(N)$ is a linear combination of $1,1.5^{-N},3^{-N},2^{-N}$, with a positive coefficient only on $1$ and negative coefficients on the other components. This implies that $f(N)$ reaches its minimum in $\{N\in\mathbb{N}_+|N>5\}$ at $N=6$,
\begin{equation}
f(N) \geq 1 - 0.0313 - 0.3277 - 0.0308>0.
\end{equation}
According to Gershgorin's circle theorem, the eigenvalues of $\mathcal{M}_{\Phi}(\tau_{P,N})$ are in the union set $\{|x-r_1|\leq R_1\}\cup\cdots\cup\{|x-r_6|\leq R_6\}$.
As $r_i-R_i>0$, this set is on the plane $\{x|\text{Re}x>0\}$. Additionally, note that $\mathcal{M}_{\Phi}(\tau_{P,N})$ is self-conjugate and therefore all of its eigenvalues are real, i.e., on the line $\{x|\text{Im}x=0\}$.
Therefore, the eigenvalues of $\mathcal{M}_{\Phi}(\tau_{P,N})$ are all positive. Then the matrix is positive definite when $x=0$ for all $N>5$.
\end{proof}
By Lemmas~{\ref{lem:case_N_leq_5} and \ref{lem:case_N_greater_than_5}}, we have proven the theorem in the main text that the squashing model exists for the standard three-state protocol.

\section{Formulation of the numerical framework}
\label{sec:pcd}
The post-selection operator in this protocol is different from the one in \cite{Winick2018ReliableNK} because Bob's announcement depends on Alice's announcement in the three-state protocol, which is not allowed in the original framework.

In the following procedure, Bob's announcement is denoted as $s$ (either "pass" or "fail") and Alice's announcement is denoted as $r$ (either 0, 1, or 2). These announcements are stored in the classical system $\tilde B$ and $\tilde A$, respectively. Additionally, Alice's measurement output $(r,b)$ is jointly denoted as $q$ ($q \in \{w_0, w_1, w_2\}$ where $w_0 = \{00, 01\}$, $w_1 = \{10, 11\}$, and $w_2 = \{20, 21\}$), and Bob's measurement output is denoted as $p$ ($p\in \{0, 1, 2\}$). These outputs are stored in the classical system $\bar A$ and $\bar B$, respectively.

In step 3 of the protocol description in the main text, Alice and Bob make measurements and announce certain information. In step 4, Bob obtains the transmission value using the key map $g(p, r)$ as defined in Table \ref{tab:step4_mapping}, and this value is stored in the quantum register $R$. Therefore, the post-selection operator can be expressed as follows,
\begin{equation}\label{equ:post_selection}
\mathcal{G}(\rho) = V\Pi\mathcal{A}_B(\mathcal{A}_A(\rho))\Pi^\dagger V^\dagger,
\end{equation}
where we slightly abuse the notation $g(p,r)=s$, where $0,1$ is taken as "pass" and $X$ is taken as "fail". The operators $\mathcal{A}_A(\rho)$ and $\mathcal{A}_B(\rho)$ are defined as
\begin{equation}
\begin{aligned}
\mathcal{A}_A(\rho) &= \sum_{r\in\{0, 1, 2\}} K^A_r \rho {K^A_r}^\dagger \\
\mathcal{A}_B(\rho) &= \sum_{s\in\{\text{pass},\text{fail}\}} K^B_s\rho {K^B_s}^\dagger,
\end{aligned}
\end{equation}
where
\begin{equation}
\begin{aligned}
K^A_r &= \sum_{q \in w_r} \ket{\psi_{f(q_0, q_1)}}\bra{\psi_{f(q_0, q_1)}}_A \otimes \ket{r}_{\tilde{A}} \otimes \ket{q}_{\bar{A}} \\
K^B_s &= \sum_{\substack{r \in \{0, 1, 2\} \\ p \in \{p | g(p, r) = s\}}} (P_p)_B \otimes \ket{r}\bra{r}_{\tilde{A}} \otimes \ket{s}_{\tilde{B}} \otimes \ket{p}_{\bar{B}},
\end{aligned}
\end{equation}
and
\begin{equation}
\begin{aligned}
\Pi &= \sum_q \ket{\psi_{f(q_0, q_1)}}\bra{\psi_{f(q_0, q_1)}}_A \otimes \ket{\text{pass}}\bra{\text{pass}}_{\tilde{B}} \\
V &= \sum_{p, r} \ket{g(p, r)}_R \otimes \ket{r}\bra{r}_{\tilde{A}} \otimes \ket{p}\bra{p}_{\bar{B}}.
\end{aligned}
\end{equation}
With these definitions, the formulation of the operator $\mathcal{G}$ is complete.

Recalling that
\begin{equation}
f(\rho) = D(\mathcal{G}(\rho_{AB})\|\mathcal{Z}(\mathcal{G}(\rho_{AB}))),
\end{equation}
where $\mathcal{G}(\rho_{AB})$ has systems $RA\tilde A\bar A B \tilde B\bar B$ and $\mathcal{Z}(\mathcal{G}(\rho_{AB}))$ has systems $Z^R A\tilde A\bar A B \tilde B\bar B$. By the monotonicity of relative entropy (where our completely positive and trace-preserving operation is the partial trace over system $\tilde A\bar A\tilde B$), we have
\begin{equation}
D(\mathcal{G}(\rho_{AB})\|\mathcal{Z}(\mathcal{G}(\rho_{AB}))) \geq D(\rho_{RAB\bar B}\| \rho_{Z^RAB\bar B}).
\end{equation}
We consider a simplified version of $\mathcal{G}$ and $\mathcal{Z}$, where system $\tilde A\bar A\tilde B$ is traced out. Explicitly, $\rho_{RAB\bar B} = \sum_{i=1,2,3}K_i\rho K_i^\dagger$, where
\begin{equation}
\begin{aligned}
K_1 &=
(\ket{000}\bra{00}+\ket{101}\bra{01})_{A\to RA}
\otimes
\left(\sqrt{\frac{2}{3}}\ket{\overline{\psi_1}}\bra{\overline{\psi_1}}\otimes\ket{0}+\sqrt{\frac{2}{3}}\ket{\overline{\psi_2}}\bra{\overline{\psi_2}}\otimes\ket{1}\right)_{B\rightarrow B\bar{B}} \\
K_2 &=
(\ket{010}\bra{10}+\ket{111}\bra{11})_{A\to RA}
\otimes
\left(\sqrt{\frac{2}{3}}\ket{\overline{\psi_2}}\bra{\overline{\psi_2}}\otimes\ket{1}+\sqrt{\frac{2}{3}}\ket{\overline{\psi_3}}\bra{\overline{\psi_3}}\otimes\ket{2}\right)_{B\rightarrow B\bar{B}} \\
K_3 &=
(\ket{020}\bra{20}+\ket{121}\bra{21})_{A\to RA}
\otimes
\left(\sqrt{\frac{2}{3}}\ket{\overline{\psi_3}}\bra{\overline{\psi_3}}\otimes\ket{2}+\sqrt{\frac{2}{3}}\ket{\overline{\psi_1}}\bra{\overline{\psi_1}}\otimes\ket{0}\right)_{B\rightarrow B\bar{B}}. \\
\end{aligned}
\end{equation}
In step 4, measurements are performed on the quantum register $R$, which stores the final result, and the classical result $Z^R$ is obtained,
\begin{equation}\label{equ:key_map}
\rho_{Z^RAB\bar B}=\sum_{\{0,1\}} (\ket{j}\bra{j}_{R\to Z^R}\otimes I)\rho_{RAB\bar B} (\ket{j}\bra{j}_{R\to Z^R}\otimes I).
\end{equation}
By using the squashing model and considering the single-photon component of the coherent state source, the objective of the coherent state case is the same as that of the single-photon case.
\end{document}